\documentclass[11pt]{article}

\usepackage{amsmath,amsthm,amsfonts,amssymb,mathrsfs,bm,bbm}
\usepackage{xspace}
\usepackage[dvipsnames]{xcolor}
\usepackage{natbib}
\usepackage{authblk}
\usepackage[margin=1in]{geometry}
\usepackage{txfonts}
\usepackage[T1]{fontenc}
\usepackage{libertine}
\usepackage[
pagebackref,
colorlinks=true,
urlcolor=blue,
linkcolor=blue,
citecolor=OliveGreen,
]{hyperref}
\usepackage[nameinlink]{cleveref}
\usepackage{enumitem}

\usepackage[suppress]{color-edits}
\addauthor{rdk}{blue}
\addauthor{ma}{red}
\addauthor{od}{green}

\newtheorem{thm}{Theorem}[section]
\newtheorem{lem}[thm]{Lemma}
\newtheorem{prop}[thm]{Proposition}
\newtheorem{cor}[thm]{Corollary}
\crefname{thm}{Theorem}{Theorems}
\crefname{lem}{Lemma}{Lemmas}

\theoremstyle{definition}
\newtheorem{defn}{Definition}[section]

\newcommand{\reals}{{\mathbb{R}}}
\newcommand{\naturals}{{\mathbb{N}}}

\newcommand{\eps}{\varepsilon}
\newcommand{\poly}{{\operatorname{poly}}}
\newcommand{\expect}{{\mathbb{E}}}
\newcommand{\vctr}[1]{\bm{\mathrm{#1}}}
\newcommand{\sgn}{{\operatorname{sgn}}}

\newcommand{\pratio}{{\mathsf{PR}}}
\newcommand{\tpratio}{{\mathsf{TPR}}}
\newcommand{\golden}{\varphi}

\newcommand{\tX}{{\tilde{X}}}
\newcommand{\tth}{{\tilde{\theta}}}

\newenvironment{lparray}%
{\begin{array}[t]{l@{\hspace{8mm}}l@{\hspace{8mm}}l}}%
{\end{array}}
\newlength{\lplb}
\title{Constrained-Order Prophet Inequalities}
\author[1]{Makis Arsenis}
\author[2]{Odysseas Drosis}
\author[1]{Robert Kleinberg}
\affil[1]{Cornell University, Ithaca, NY 14853, USA.\protect\\
{\footnotesize \texttt{\{marsenis,rdk\}@cs.cornell.edu}}}

\affil[2]{EPFL, 1015 Lausanne, Switzerland.\protect\\
{\footnotesize \texttt{odysseas.drosis@epfl.ch}}}
\date{}

\begin{document}

\begin{titlepage}
\maketitle

\begin{abstract}
  Free order prophet inequalities bound the ratio between
  the expected value obtained by two parties each selecting
  one value from a set of independent random variables:
  a ``prophet'' who knows the value of each variable and
  may select the maximum one, and a ``gambler'' who is free to
  choose the order in which to observe the values but must
  select one of them immediately after observing it, without
  knowing what values will be sampled for the unobserved
  variables. It is known that the gambler can always
  ensure an expected payoff at least 0.669\dots times
  as great as that of the prophet. In fact, even if the
  gambler uses a threshold stopping rule, meaning there
  is a fixed threshold value such that the gambler rejects
  every sample below the threshold and accepts every sample
  above it, the threshold can always be chosen so that
  the gambler-to-prophet ratio is at least
  $1-\frac1e \approx 0.632\ldots$. In contrast, if
  the gambler must observe the values in a predetermined
  order, the tight bound for the
  gambler-to-prophet ratio is $1/2$.

  In this work we investigate a model that interpolates between
  these two extremes. We assume there is a predefined set of
  permutations of the set indexing the random variables, and the
  gambler is free to choose the order of observation to be any
  one of these predefined permutations. Surprisingly, we show
  that even when only two orderings are allowed --- namely,
  the forward and reverse orderings --- the gambler-to-prophet
  ratio improves to $\golden^{-1} = \frac12 (\sqrt{5} - 1) = 0.618\ldots$,
  the inverse of the golden ratio.
  As the number of allowed permutations grows
  beyond 2, a
  striking ``double plateau'' phenomenon emerges:
  after increasing from $0.5$ to $\golden^{-1}$ when
  two permutations are allowed, the gambler-to-prophet
  ratio achievable by threshold stopping rules
  does not exceed $\golden^{-1} + o(1)$ until the
  number of allowed permutations grows to $O(\log n)$.
  The ratio reaches $1 - \frac1e - \eps$ for a suitably
  chosen set of $O(\poly(\eps^{-1}) \cdot \log n)$ permutations
  and does not exceed $1 - \frac1e$ even when the full set
  of $n!$ permutations is allowed.
\end{abstract}

\end{titlepage}

\section{Introduction}

When one is planning under uncertainty,
the order in which decisions must be made
can powerfully influence the quality of
the eventual outcome. It is therefore
advantageous for a decision maker to be
able to control the order of the decisions
they will face. For example,
a new Ph.D.\ student might prefer to
know if they will be offered a position in
the research group of their top-choice advisor
before having to decide whether to accept
a position with their second-choice advisor,
but they might not have the freedom to
dictate the order in which those two decisions
are made.

Optimal stopping theory furnishes a theoretical
basis for quantifying the benefit of exercising
control over the order of one's decisions in
uncertain environments. In the classical optimal
stopping problem there is a sequence of independent random
variables $X_1,\ldots,X_n$ with known distributions,
and one aims to design and analyze a stopping rule
$\tau$ that maximizes the expected value of the sample
selected by $\tau$, namely $\expect X_\tau$. Optimal stopping
with order selection, introduced by \citet{hill83},
posits that the decision maker can choose the order
in which to observe the random variables {\em and}
the time at which to stop the permuted sequence.
In other words, the aim is now to design a
permutation $\pi : [n] \to [n]$ and a
stopping rule $\tau$ adapted to the sequence
$X_{\pi(1)},\ldots,X_{\pi(n)}$, such
that $\expect X_{\tau}$ is maximized.
(In principle one could also consider a model
in which the permutation $\pi$ is specified
adaptively, i.e.~the value $\pi(k)$ is
allowed to depend on $X_{\pi(1)},\ldots,X_{\pi(k-1)}$.
One of the main results of \citet{hill83}, Theorem 3.11,
shows that the optimal adaptive ordering is no better
than the optimal non-adaptive one.)

Prophet inequalities, which bound the ratio between
the expected value selected by an optimal stopping rule
(metaphorically, a ``gambler'')
and the expectation of the maximum value in the
sequence (metaphorically, a ``prophet''),
are a fruitful way of analyzing the
effectiveness of different types of stopping rules.
For optimal stopping without order selection,
\citet{KS77,KS78} famously proved that the gambler's
expected value is always at least one-half that of
the prophet, and the ratio 1/2 in this bound is the
best possible.
\citet{S-C} showed that the optimal ratio 1/2 remains
attainable if the gambler is constrained to use a
{\em threshold stopping rule}, which sets
a fixed threshold $\theta$ and commits to reject
every sample less than $\theta$ and to stop
whenever it encounters a sample strictly
greater than $\theta$.

For the optimal stopping problem with order selection,
bounds on the ratio between the gambler's and prophet's
expected values are known as {\em free order prophet
inequalities}. The first free order prophet inequality
was proven by \citet{yan-free-order}, who showed that
the gambler-to-prophet ratio is always at least
$1 - \frac1e = 0.632\ldots$.
This bound was later shown to be attainable
even if the gambler is constrained to observe the values in
uniformly-random order \citep{proph-sec} and to use a
threshold stopping rule \citep{blind-strat}\footnote{%
  To achieve gambler-to-prophet ratio $1-\frac1e$
  using a threshold stopping rule, it is vital that
  our definition of threshold stopping rule treats
  its behavior at time $t$ as being unconstrained
  if the value $X_t$ is equal to the threshold $\theta$;
  the requirement of stopping at time $t$ is only triggered
  if $X_t > \theta$. For a stricter definition of
  threshold stopping rule that requires stopping
  whenever $X_t \geq \theta$ and continuing otherwise,
  \citet{proph-sec} showed that 1/2 is the best
  achievable gambler-to-prophet ratio.
}. Furthermore,
$1 - \frac1e$ is asymptotically
the best possible ratio attainable by threshold stopping
rules \citep{delegated-search}, even if the distributions of
$X_1,\ldots,X_n$ are identical. However, general stopping
rules can do strictly better: the optimal factor in the
free-order prophet inequality is known to be between $0.669\ldots$
and $0.745\ldots$, and closing the gap between
these two bounds is a major open question.

\subsection{Our contributions}

To summarize the foregoing discussion,
the gap between the gambler-to-prophet
ratios attainable with or without order selection
formalizes, and quantifies, the advantage that a
decision maker gains by being able to control the
order in which decisions are made under uncertainty.
But how much control over the ordering is needed
to gain this advantage? Our paper initiates
an in-depth exploration of that question.
We introduce
{\em constrained-order prophet inequalities},
in which there is a predefined set $\Pi \subseteq S_n$
of permutations of $[n]$, and the gambler is allowed to reorder the
random variables $X_1,\ldots,X_n$ using any permutation
in $\Pi$ before running a stopping rule.
\begin{defn} \label{def:prophet}
  A non-empty set of permutations
  $\Pi \subseteq S_n$ is said to satisfy a
  {\em constrained-order prophet inequality with factor $\alpha$}
  if for {\em every} $n$-tuple of distributions
  $X_1,\ldots,X_n$ supported on the non-negative reals,
  there is a permutation $\pi \in \Pi$ and a stopping
  rule $\tau$ adapted to $X_{\pi(1)},X_{\pi(2)},\ldots,X_{\pi(n)}$,
  such that
  \begin{equation} \label{eq:copi}
    \expect X_\tau \geq \alpha \cdot
    \expect \left[ \max_{1 \leq i \leq n} X_i \right] .
  \end{equation}
  The {\em prophet ratio} of $\Pi$,
  $\pratio(\Pi)$, is the supremum of all $\alpha$
  such that $\Pi$ satisfies a constrained-order prophet
  inequality with factor $\alpha$.
  {\em Constrained-order threshold
  prophet inequalities} and
  the {\em threshold
  prophet ratio} $\tpratio(\Pi)$ are defined similarly,
  but allowing the gambler to optimize only over threshold
  stopping rules rather than all stopping rules.
\end{defn}
In the extreme cases where $\Pi$ has only one
element or $\Pi$ is the entire permutation
group $S_n$, one recovers the definitions of
prophet inequality and free-order prophet inequality,
respectively.
Constrained-order prophet inequalities
interpolate between these two extremes
and allow
us to gain deeper insight into how and why
optimizing the order of decisions leads to
better outcomes for optimal stopping rules.
Interestingly, our first main result shows
that much of the benefit of order selection
can be gained by choosing the better of just
{\em two} permutations, corresponding
to the forward and reverse orderings.
\begin{thm} \label{thm:fwd-rev}
  Let $\Pi = \{\iota,\rho\}$, where $\iota$
  and $\rho$ are the permutations of $[n]$
  that place its elements in forward and
  reverse order, respectively.
  The threshold prophet ratio of $\Pi$
  is
  $  \tpratio(\Pi) = \golden^{-1} = \tfrac12
     \left( \sqrt{5} - 1 \right) = 0.618\ldots .
  $
\end{thm}
One may interpret this result as lending
support to the intuition that most of the
benefit of order selection stems from the
ability to schedule a high-risk high-reward
option (e.g.~a value that is
$1/\eps$ with probability $\eps$, else
zero) in the first half of the decision
sequence, leaving safer options for afterward.

Given the large quantitative gain
in $\tpratio(\Pi)$ when $\Pi$ is
enlarged to contain the reverse ordering
of the input sequence,
one might expect to extract
additional significant gains
when $\Pi$ is allowed to contain
three permutations, or an even
larger constant number of them.
Surprisingly, we show this is not the
case: to exceed the golden-ratio
bound at all one needs a super-constant
number of permutations, and to exceed
it by any constant $\eps>0$
one needs a logarithmic number of
them.
\begin{thm} \label{thm:golden-ratio}
  If $n \ge 3$ and $|\Pi| < \sqrt{\log n}$
  then $\tpratio(\Pi) \le \golden^{-1}$.
  For any $\eps>0$, if $|\Pi| < \log_{1/\eps}(n)$
  then $\tpratio(\Pi) \le \golden^{-1} + O(\eps)$.
\end{thm}
Recall that if the gambler is allowed to
order the elements arbitrarily, the best possible
threshold prophet ratio is $\tpratio(S_n) = 1 - \frac1e + o(1)$.
Our third and final main result shows that a logarithmic
number of permutations are sufficient to get within
$\eps$ of this bound, and that a quadratic number of
permutations suffice to match the $1 - \frac1e$ bound
exactly.
\begin{thm} \label{thm:pseudorandom}
  For every $n \in \mathbb{N}$ and
  $\eps>0$, there is a set $\Pi$ consisting
  of $O(\poly(\eps^{-1}) \cdot \log n)$ permutations
  such that $\tpratio(\Pi) > 1 - \frac1e - \eps$.
  There is also a set $\Pi$ consisting of
  $O(n^2)$ permutations such that $\tpratio(\Pi) \geq
  1 - \frac1e$.
\end{thm}
Taken together, our results give a nearly complete
answer to the question: for a given $\alpha$,
what is the smallest $m$ such that there
exists an $m$-element set of permutations
whose threshold prophet ratio is at least
$\alpha$? The answer is:
\begin{enumerate}
  \item $m=1$ for $\alpha \leq \frac12$ \citep{KS77,KS78};
  \item $m=2$ for $\frac12 < \alpha \leq \golden^{-1}$
    (\Cref{thm:fwd-rev,thm:golden-ratio});
  \item $m = \Theta(\log n)$ for $\golden^{-1} < \alpha < 1 - \frac1e$,
    with the constant inside the $\Theta(\cdot)$ depending
    on the value of $\alpha$ (\Cref{thm:pseudorandom});
  \item $m = O(n^2)$ for $\alpha = 1 - \frac1e$ (\Cref{thm:pseudorandom});
  \item there is no set of permutations whose
  threshold prophet ratio is greater than $1 - \frac1e$
  \citep{delegated-search}; see also \Cref{prop:1-1/e-optimal} below.
\end{enumerate}
It would be desirable, of course, to gain
a similarly comprehensive understanding of
the smallest set of permutations needed to
achieve a given prophet ratio (rather than
threshold prophet ratio), but at present
this seems out of reach: even determining
the prophet ratio of the full permutation
group, $\pratio(S_n)$, is a major open problem
as it amounts to determining the best possible
constant in the free order prophet inequality.
For now, the most we can say is that
our results imply lower bounds on the
prophet ratio $\pratio(\Pi)$ for the
permutation sets $\Pi$ we study,
due to the trivial observation that
$\pratio(\Pi) \geq \tpratio(\Pi)$ for
every set $\Pi$.


In summary, then, the message of our paper
is as follows: a gambler solving an
optimal stopping problem
using a threshold stopping rule can
achieve significant gains if they
are allowed to control the order in
which they observe the values in the
sequence, but most of these gains can
be achieved just by choosing the better
of the forward or reverse ordering, and
nearly all of the gains can be
achieved by choosing among a logarithmic
number of predefined permutations.

\subsection{Relation to the prophet secretary problem}

Although we have motivated and presented
our results in terms of free-order and
constrained-order prophet inequalities,
i.e.~a setting in which the gambler chooses
the order in which values are observed,
our results also have a bearing on the
prophet secretary problem due to their
method of proof. In the prophet secretary
problem~\citep{proph-sec,azar18ec,prophet-matroid-sec}
the gambler observes $n$ independent random variables
in {\em uniformly random order}, and once
again the question is what gambler-to-prophet
ratio can be guaranteed. The best bounds
currently known, due to \citet{blind-strat},
show that the answer is at least $0.669\ldots$
and at most $\sqrt{3}-1 = 0.732\ldots$.

The constrained-order prophet inequalities
asserted in \Cref{thm:fwd-rev,thm:golden-ratio,thm:pseudorandom}
are all proven by constructing a small
set of permutations, $\Pi$, and analyzing
the performance of a threshold stopping
rule when the order in which values are
observed is drawn uniformly at random
from $\Pi$. Thus, all of our results can
also be interpreted as constructing
``pseudo-random'' distributions over
permutations that have small support
size, but ensure a gambler-to-prophet ratio
(for threshold stopping rules)
that is nearly as good as what can
be achieved in the prophet secretary
problem when the order of observation is
sampled uniformly at random.

\section{Preliminaries}
\label{sec:prelim}

This section presents definitions, notations,
and conventions that will be used throughout the paper.

Let $X_1,X_2,\ldots,X_n$ be any $n$-tuple of
independent random variables supported on the
non-negative reals.
A {\em stopping rule} adapted to the
sequence $X_1,\ldots,X_n$ is a random
variable $\tau$ taking values in $[n] \cup \{\bot\}$,
such that for all $i \in [n]$
the event $\{ \tau = i \}$ is
measurable with respect to the
$\sigma$-field generated by
$X_1,\ldots,X_i$. One interprets
$\tau$ as defining the time at
which a gambler selects a value
from the sequence --- with $\tau = \bot$
denoting the event that no value is
selected --- and the measurability
condition expresses the notion that
the gambler must decide whether or not
to select $i$ after having seen only
the values $X_1,\ldots,X_i$.
For a stopping rule $\tau$, the
random variable $X_\tau$ is defined
as follows: if $\tau = i \in [n]$
then $X_\tau = X_i$, and if $\tau = \bot$
then $X_\tau = 0$.
If $\pi : [n] \to [n]$ is a permutation,
a stopping rule $\tau$ is called
{\em $\pi$-adapted} if it is adapted
to the sequence $X_{\pi(1)},\ldots,X_{\pi(n)}$,
in other words, the event $\{\tau = \pi(i)\}$
is measurable with respect to the $\sigma$-field
generated by $X_{\pi(1)},X_{\pi(2)},\ldots,X_{\pi(i)}$.

A stopping rule $\tau$ is a {\em threshold
stopping rule} with threshold $\theta$
if it never selects a value strictly less
than $\theta$, and it always selects the
first value strictly greater than $\theta$.
In other words, $\tau$ must satisfy the
following constraints for all $i \in [n]$:
\[
  X_i < \theta \Rightarrow \tau \neq i, \qquad
  X_i > \theta \Rightarrow \tau \leq i.
\]
Note that the tie-breaking behavior of a
threshold stopping rule is unconstrained:
when $X_i = \theta$, then $\tau = i$ is
allowed but not required\footnote{As noted
in the introduction, this tie-breaking
convention is not universal. For example
\citet{proph-sec} adopt the stricter convention
that $X_i = \theta$ implies $\tau = i$.}.
If the distributions of $X_1,\ldots,X_n$
have no point-masses, the event $X_i = \theta$
has probability zero so the tie-breaking
convention is inconsequential. To reduce
from the general case to the case where
tie-breaking is inconsequential, we adopt
the following artificial but convenient
convention. We assume that in addition
to the random variables $X_1,\ldots,X_n$,
there is an auxiliary sequence of random
variables $\tX_1,\ldots,\tX_n$, each
uniformly distributed in $[0,1]$,
independent of $X_1,\ldots,X_n$
and mutually independent of one another.
These auxiliary values are used for
tie-breaking as follows. The pairs
$\{ (X_i,\tX_i) \}_{i=1}^n$ are
regarded as elements of $\reals \times [0,1]$
under the lexicographic ordering.
For every $\tth \in [0,1]$, the
threshold stopping rule with
threshold $(\theta,\tth)$ is
defined as
  \[ \tau = \min \{ i \in [n] \mid (X_i, \tX_i) \ge (\theta,\tth) \} \]
where, again, the relation $<$ is
interpreted lexicographically.
We make the following observations.
\begin{enumerate}
  \item For all $i \in [n]$, the event $(X_i,\tX_i) = (\theta,\tth)$
    has probability zero.
  \item For any $i \in [n]$ and $p \in (0,1),$
    we can find $(\theta,\tth)$ such that
    $\Pr((X_i,\tX_i) > (\theta,\tth)) = p$.
  \item Similarly, if $(X_*,\tX_*)$ denotes the
    $(\reals \times [0,1])$-valued random variable
    that is the lexicographic maximum of
    $\{(X_i,\tX_i)\}_{i=1}^n$,
    then for any $p \in (0,1)$ we can find
    $(\theta,\tth)$
    such that $Pr((X_*,\tX_*) > (\theta,\tth)) = p$.
\end{enumerate}
In short, by treating the sequence of $n$ random
variables as taking values in $\reals \times [0,1]$
(lexicographically ordered) and treating the
threshold as belonging to the same set, we can
make the same continuity and no-tie-breaking
assumptions that are always justified in the
case of point-mass-free distributions, but
without having to assume the distributions
of $X_1,\ldots,X_n$ have no point masses.

To avoid cumbersome notation in what follows,
when discussing threshold stopping rules we
will still denote the threshold by $\theta$
rather than $(\theta,\tth)$ and we'll use the
notation $X_i < \theta$ or $X_* < \theta$
as shorthand for $(X_i,\tX_i) < (\theta,\tth)$
or $(X_*,\tX_*) < (\theta,\tth)$. In short,
we'll treat the distributions of $X_1,\ldots,X_n$
as if they were free of point masses, depending
on the conventions set forth in this section
to justify that the results derived under the
point-mass-free assumption extend to the case
of general distributions.

For a threshold $(\theta,\tth)$ and a permutation
$\pi$, the stopping rule $\tau(\pi,\theta,\tth)$
selects the earliest element of the
sequence $X_{\pi(1)},\ldots,X_{\pi(n)}$ such that
$(X_{\pi(i)},\tX_{\pi(i)}) \geq (\theta,\tth)$.
More precisely, define $\tau(\pi,\theta,\tth)$
as follows. If $(X_i,\tX_i) < (\theta,\tth)$
for all $i \in [n]$ then $\tau(\pi,\theta,\tth) = \bot$.
Otherwise, $\tau(\pi,\theta,\tth) = \pi(i_{\min})$
where $i_{\min}$ is the minimum $i$ such that
$(X_{\pi(i)},\tX_{\pi(i)}) \geq (\theta,\tth)$.
We will use the notation $X_{\pi,\theta}$ as shorthand
for $X_{\tau(\pi,\theta,\tth)}$.

The set of all permutations of $[n]$ is
denoted by $S_n$.
Two permutations that are important
in this work are the identity permutation
$\iota(k) = k$ and its reverse,
$\rho(k) = n+1-k$.
Recall that for a non-empty
subset $\Pi \subseteq S_n$, the terms
{\em constrained-order (threshold) prophet inequality}
and {\em (threshold) prophet ratio} were
defined above, in \Cref{def:prophet}.
The threshold prophet ratio of $\Pi$ is
denoted by $\tpratio(\Pi)$.

\section{Using the forward and reverse permutations}
\label{sec:fwd-rev}

In this section, we let $\Pi = \{\iota,\rho\}$
and prove that $\tpratio(\Pi) = \golden$.

\begin{thm} \label{thm:golden-is-achievable}
  For every $n$-tuple of independent random
  variables $X_1,\ldots,X_n$, there exists a
  threshold $\theta$ such that
  \begin{equation} \label{eq:golden-is-achievable}
    \expect_{\pi} \left[ \expect X_{\pi,\theta} \right]
    \geq
    \golden^{-1} \cdot \expect X_* ,
  \end{equation}
  where the outer expectation on the left side is over
  a uniformly random choice of $\pi \in \Pi = \{\iota,\rho\}$.
\end{thm}
\begin{proof}
  For a given threshold $\theta$, let $p = \Pr(X_* \ge \theta)$.
  We have
  \begin{equation} \label{eq:gia.1}
    \expect X_* \le \theta + \expect[(X_*-\theta)^+] \le
      \theta + \sum_{i=1}^n \expect[(X_i - \theta)^+] .
  \end{equation}
  For the random variable $X_\tau = X_{\pi,\theta}$, we have
  \begin{equation} \label{eq:gia.2}
    \expect X_\tau = p \theta + \expect[(X_\tau - \theta)^+]
      = p \theta + \sum_{i=1}^n c_i \expect[(X_i - \theta)^+]
  \end{equation}
  where
  \begin{equation} \label{eq:yi}
    c_i = \frac12 \left(
      \prod_{j=1}^{i-1} \Pr(X_j < \theta) +
      \prod_{j=i+1}^n \Pr(X_j < \theta)
    \right)
  \end{equation}
  denotes the probability that no element is
  selected before the stopping rule observes
  $X_i$, given that the sequence is observed
  in forward or reverse order with equal probability.
  Letting
  \begin{align*}
    a_i  = \prod_{j=1}^{i-1} \Pr(X_j < \theta), & \qquad
    b_i  = \prod_{j=i+1}^n \Pr(X_j < \theta),
  \end{align*}
  we have
  \begin{align}
    a_i b_i & \geq \prod_{j=1}^n \Pr(X_j < \theta) = 1-p  \\
    c_i & = \tfrac12 (a_i + b_i) \geq (a_i b_i)^{1/2} \geq \sqrt{1-p} ,
    \label{eq:gia.3}
  \end{align}
  where the second line follows from the arithmetic-geometric
  mean inequality. Letting $q = \sqrt{1-p}$ and
  substituting $c_i \geq q$ and $p = 1-q^2$
  back into Eq.~\eqref{eq:gia.2} we obtain
  \begin{equation} \label{eq:gia.4}
    \expect X_\tau \geq
       (1 - q^2) \theta + q \sum_{i=1}^n \expect[(X_i - \theta)^+] .
  \end{equation}
  Choosing $\theta$ so that $q = \golden^{-1}$, which implies also
  $1 - q^2 = q = \golden^{-1}$, we find that
  \begin{equation} \label{eq:gia.5}
    \expect X_\tau \geq \golden^{-1} \cdot \left( \theta + \sum_{i=1}^n \expect[(X_i - \theta)^+] \right)
      \geq \golden^{-1} \cdot \expect X_*,
  \end{equation}
  where the second inequality follows from~\eqref{eq:gia.1}.
\end{proof}

Next we prove that the golden-ratio bound in \Cref{thm:golden-is-achievable}
is the best possible.

\begin{lem} \label{lem:3-distribs}
  For all $\eps > 0$
  there exists a sequence of three independent random variables $X_1,X_2,X_3$
  such that $X_1$ and $X_3$ are identically distributed, and for every threshold
  stopping rule $\tau$ we have
  \begin{equation} \label{eq:3-distribs}
    \expect X_\tau < ( \golden^{-1} + \eps ) \cdot \expect X_* .
  \end{equation}
\end{lem}
\begin{proof}
  Fix a parameter $\delta > 0$ to be determined later, and suppose $X_1$ and $X_3$
  are uniformly distributed in $[1 - \delta, 1]$ while $X_2$ is equal to
  $(\sqrt{5}-1)/\delta$ with probability $\delta$, otherwise $X_2=0$.
  Then
  \begin{equation} \label{eq:3d.1}
    \expect X_* > \delta \cdot (\sqrt{5}-1)/\delta \, + \, (1 - \delta) \cdot (1 - \delta)
        > \sqrt{5} - 2 \delta.
  \end{equation}
  If $\tau$ is a threshold stop rule
  with threshold $\theta$ there are three cases to consider.
  When $\theta < 1 - \delta$, we have $\expect X_\tau = 1 - \frac{\delta}{2}$.
  When $\theta > 1$ we have $\expect X_\tau = \sqrt{5} - 1$.
  In both of these cases, $\expect X_\tau < (\golden^{-1} + \eps) \cdot \expect X_*$
  provided $\delta$ is sufficiently small.
  The remaining case is when $1 - \delta \leq \theta \leq 1$. In this case,
  let $r = (1 - \theta) / \delta$, so that $\Pr(X_1 > \theta) = r$.
  Then
  \begin{align} \nonumber
    \expect X_\tau & \le r \cdot 1 \, + \,
                       (1-r) \delta \cdot (\sqrt{5}-1)/\delta \, + \,
                       (1-r)(1-\delta)r \cdot 1 \\
                   & \le r + (1-r)(\sqrt{5}-1) + (r - r^2)
                   = (\sqrt{5} - 1) + (3 - \sqrt{5}) r - r^2 .
    \label{eq:3d.2}
  \end{align}
  The right side of~\eqref{eq:3d.2} is maximized when $r =
  \frac12 ( 3 - \sqrt{5} )$, when it equals $(\sqrt{5} - 1) + \frac14 (14 - 6 \sqrt{5}) = \frac12 (5 - \sqrt{5})$.
  Hence,
  \begin{equation} \label{eq:3d.3}
    \expect X_\tau \le \frac12 (5 - \sqrt{5}) = \golden^{-1} \cdot \sqrt{5}.
  \end{equation}
  Combining~\eqref{eq:3d.1} with~\eqref{eq:3d.3} we see that
  the conclusion of the lemma holds, as long as $\delta$
  is chosen small enough that $(\golden^{-1} + \eps) \cdot (\sqrt{5} - 2 \delta) \geq \golden^{-1} \cdot \sqrt{5}$.
\end{proof}

\begin{cor} \label{cor:3-indices}
  Suppose $\Pi$ is a non-empty set of permutations of $[n]$
  and
      $i,j,k \in [n]$ are three distinct indices such
      that $\pi^{-1}(j)$ is between $\pi^{-1}(i)$ and $\pi^{-1}(k)$
      for all $\pi \in \Pi$.
  Then $\tpratio(\Pi) \leq \golden^{-1}$.
\end{cor}
\begin{proof}
  Define an $n$-tuple of independent random variables
  $X_1,\ldots,X_n$ by specifying that the distributions
  of $X_{i},X_{j},X_{k}$ are identical to the
  distributions of $X_1,X_2,X_3$ specified
  in \Cref{lem:3-distribs} above, and for $\ell \not\in \{i,j,k\}$
  let the distribution of $X_\ell$ be identically zero.
  For any $\pi \in \Pi$, a $\pi$-adapted
  threshold stopping rule with threshold $\theta > 0$
  will skip $X_{\ell}$ for every $\ell \not\in \{i,j,k\}$,
  and it will observe the values $X_i,X_j,X_k$ at times
  $\pi^{-1}(i), \pi^{-1}(j), \pi^{-1}(k)$. Note that the
  distributions of these three variables and the order
  in which they are observed are identical to the
  distributions $X_1,X_2,X_3$ specified in
  \Cref{lem:3-distribs}, so the inequality~\eqref{eq:3-distribs}
  will be satisfied. As $\tau$ is an arbitrary threshold
  stopping rule, and $\eps > 0$ is arbitrarily small, we
  conclude that $\Pi$ does not satisfy a constrained-order
  prophet inequality with factor $\alpha$ for any
  $\alpha > \golden^{-1}$.
\end{proof}

\begin{thm} \label{thm:iotarho}
  For the permutation set $\Pi = \{\iota,\rho\}$,
  if $n \geq 3$, we have $\tpratio(\Pi) = \golden^{-1}$.
\end{thm}
\begin{proof}
  A threshold stopping rule that is allowed to select the better
  of $\iota$ and $\rho$ can do no worse than one which samples
  one of the two permutations uniformly at random. Therefore,
  \Cref{thm:golden-is-achievable} implies $\tpratio(\Pi) \geq \golden^{-1}$.
  On the other hand, any three distinct indices $i < j < k$
  in $[n]$ satisfy the condition in \Cref{cor:3-indices},
  implying that $\tpratio(\Pi) \leq \golden^{-1}$.
\end{proof}

\section{Beating the golden-ratio bound requires many permutations}
\label{sec:beating}

In this section we show that $\tpratio(\Pi) \leq \golden^{-1}$
whenever $|\Pi| < \sqrt{\log n}$ and that
$\tpratio(\Pi) \leq \golden^{-1} + O(\eps)$
whenever $\eps > 0$ and $|\Pi| < \log_{1/\eps}(n)$.
This was stated as \Cref{thm:golden-ratio} in the
introduction, and is proven by combining
\Cref{lem:golden-is-optimal,lem:centered} below.

\Cref{cor:3-indices} presented a necessary condition
for a permutation set $\Pi$ to satisfy
$\tpratio(\Pi) > \golden^{-1}$: for every
three indices $i,j,k$ there exists $\pi \in \Pi$
such that $\pi^{-1}(j)$ does not
lie between $\pi^{-1}(i)$ and $\pi^{-1}(k)$.
It turns out this condition guarantees
$|\Pi| > \log \log n$ \citep{nonunif}
but it does not imply any stronger lower bound
than that doubly-logarithmic one. To prove
stronger lower bounds on $|\Pi|$, we begin
by generalizing \Cref{cor:3-indices}.

\begin{defn} \label{def:centered}
  If $\Pi$ is a set of permutations of $[n]$,
  we say an index $j \in [n]$
  is {\em $\eps$-centered with
  respect to $\Pi$}
  if there exists a probability distribution $p$
  on $[n] \setminus \{j\}$ such that for every
  $\pi \in \Pi$ with inverse permutation $\sigma = \pi^{-1}$,
  the sets $\{ i \mid \sigma(i) < \sigma(j) \}$ and
  $\{ i \mid \sigma(i) > \sigma(j) \}$ both have
  measure greater than $\frac12 - \eps$ under $p$.
  In the special case $\eps=0$, we shall say that
  $j$ is {\em centered with respect to $\Pi$}.
\end{defn}

\begin{lem} \label{lem:golden-is-optimal}
  If $\Pi$ is a non-empty set of permutations of $[n]$
  and there exists an index $j \in [n]$ that is
  $\eps$-centered with respect to $\Pi$, then
  $\tpratio(\Pi) \leq \golden^{-1} + O(\eps)$.
\end{lem}
The proof, which is deferred to \Cref{sec:beating-defer},
generalizes the proof of \Cref{cor:3-indices}.
The value of $X_j$ is defined to be
$(\sqrt{5}-1)/\delta$ with probability $\delta$,
else $X_j=0$, whereas the remaining distributions
are all tightly concentrated around 1. Threshold
stopping rules face a dilemma: if the threshold is
set high enough that there is a reasonable probability
of not stopping before $X_j$, then there
must be a reasonable probability that the gambler also
doesn't stop after $X_j$ and comes away empty-handed.
The prophet faces no such dilemma: she can pick $X_j$
when it is non-zero, and otherwise she can pick a value
that is nearly equal to 1.

\begin{lem} \label{lem:centered}
  If $\Pi$ is a set of fewer than $\sqrt{\log n}$
  permutations of $[n]$ then there exists an index
  $j$ that is centered with respect to $\Pi$.
  If $\Pi$ is a set of fewer than $\log_{1/\eps}(n)$
  permutations of $[n]$ for some $\eps > 0$
  then there exists an index
  $j$ that is $\eps$-centered with respect to $\Pi$.
\end{lem}
\begin{proof}
  Let $\sigma_1,\sigma_2,\ldots,\sigma_m$ be an
  enumeration of the permutations whose inverse
  belongs to $\Pi$. For any pair of distinct
  indices $i \neq j$ in $[n]$ define a vector
  $\vctr{v}_{ij} \in \{ \pm 1 \}^m$ by specifying
  that for all $k \in [m]$
  \begin{equation} \label{eq:vij}
    (\vctr{v}_{ij})_k = \begin{cases}
      -1 & \mbox{if } \sigma_k(i) < \sigma_k(j) \\
      1 & \mbox{if } \sigma_k(i) > \sigma_k(j) .
    \end{cases}
  \end{equation}
  For any probability distribution $p$ on
  $[n] \setminus j$, if we use $q_k(p)$ to denote the
  measure of the set $\{ i \mid \sigma_k(i) < \sigma_k(j) \}$
  under $p$ and $\vctr{q}(p)$ to denote the vector $(q_1(p),\ldots,q_m(p))$,
  then we have
  \begin{equation} \label{eq:1-2q}
    \sum_{i \neq j} p(i) \vctr{v}_{ij} = \vctr{1} - 2 \vctr{q}(p)
  \end{equation}
  where $\vctr{1}$ denotes the vector $(1,1,\ldots,1)$.
  The criterion that $j$ is centered with respect to $\Pi$
  is equivalent to the existence of a distribution $p$
  such that $\vctr{q}(p) = \frac12 \cdot \vctr{1}$.
  Similarly, $j$ is $\eps$-centered with respect to $\Pi$
  if and only if there is a distribution $p$ such that
  $\vctr{q}(p) \in \left( \frac12 - \eps, \frac12 + \eps \right)^m$.
  Using~\eqref{eq:1-2q} we now see that
  {\em $j$ is centered with respect to $\Pi$
  if and only if $\vctr{0}$~is a convex combination
  of the vectors in the set $V_j = \{ \vctr{v}_{ij} \mid i \neq j \}$.
  Similarly, $j$ is $\eps$-centered with respect to $\Pi$
  if and only if the convex hull of $V_j$ intersects the
  hypercube $(-2 \eps, 2 \eps)^m$.}

  To finish proving the first part of the lemma,
  we must show that if $m < \sqrt{\log n}$,
  then for some index $j$, $\vctr{0}$~is a convex
  combination of the vectors in $V_j$.
  Contrapositively, we will assume that for every $j$,
  $\vctr{0}$~is not in the convex hull of
  $V_j$, and we will deduce from this assumption that
  $m^2 \geq \log n$.
  By the separating hyperplane theorem, our assumption
  that $\vctr{0}$ is not in the convex hull of $V_j$
  implies there is a vector $\vctr{w}_j$ such that
  $\langle \vctr{w}_j, \vctr{v}_{ij} \rangle > 0$
  for all $i \neq j$.
  Note that $\vctr{v}_{ij} = - \vctr{v}_{ji}$,
  so for all $i \neq j$,
  \begin{equation} \label{eq:ltf-differ}
    \langle \vctr{w}_j, \vctr{v}_{ij} \rangle > 0 >
    \langle \vctr{w}_i, \vctr{v}_{ij} \rangle ,
  \end{equation}
  where the second inequality follows because
  $\langle \vctr{w}_i, \vctr{v}_{ji} \rangle > 0$
  by our assumption on $\vctr{w}_i$.
  Now consider the linear threshold functions
  defined by $f_i(\vctr{x}) = \sgn(\langle \vctr{w}_i, \vctr{x} \rangle)$
  for each $i \in [n]$.
  Equation~\eqref{eq:ltf-differ} says that
  $f_i(\vctr{v}_{ij}) \neq f_j(\vctr{v}_{ij})$.
  Recalling that $\vctr{v}_{ij} \in \{ \pm 1 \}^m$,
  this means the restrictions of
  $f_1,\ldots,f_n$ to $\{ \pm 1 \}^m$
  are pairwise distinct.
  There are fewer than $2^{m^2}$ distinct linear
  threshold functions on $\{ \pm 1 \}^m$ \citep{cover1965geometrical},
  hence $m^2 \geq \log n$.

  To finish proving the second part of the lemma,
  we must show that if $m < \log_{1/\eps}(n)$,
  then for some index $j$, the convex hull of $V_j$
  and the hypercube $(-2 \eps, 2 \eps)^m$ intersect.
  Contrapositively, we will assume that for every
  $j$, the minimum $\infty$-norm of the vectors in the
  convex hull of $V_j$ is at least $2 \eps$. From this
  assumption we will deduce that $(1/\eps)^m \geq n$.
  Minimizing the $\infty$-norm of vectors in the
  convex hull of $V_j$ is equivalent to solving the
  following linear program, whose dual is presented
  alongside it.
  \begin{equation*}
    \begin{lparray}
      \min & r \\
      \mbox{s.t.}
      &  r - \sum_{i \neq j} v_{ij,k} p_i \geq 0
      & \forall k \in [m] \\
      &  r + \sum_{i \neq j} v_{ij,k} p_i \geq 0
      & \forall k \in [m] \\
      &  \sum_{i \neq j} p_i = 1 & \\
      &  p_i \geq 0 & \forall i \in [n] \setminus j
    \end{lparray}
    \quad
    \begin{lparray}
      \max & z \\
      \mbox{s.t.}
      & z + \sum_{k=1}^m v_{ij,k} (y_k - x_k) \leq 0
      & \forall i \in [n] \setminus j \\
      & \sum_{k=1}^m (y_k + x_k) = 1 & \\
      & x_k, y_k \geq 0 & \forall k \in [m]
    \end{lparray}
  \end{equation*}
  Our assumption is that for each $j$, the optimum
  of the primal LP is at least $2 \eps$, which
  means that the optimum of the dual LP is also
  at least $2 \eps$. Let $\vctr{x}, \vctr{y}, z$
  denote a feasible dual solution with $z \geq 2 \eps$,
  and let $\vctr{w}_j = \vctr{x} - \vctr{y}$.
  The first dual constraint, combined with the
  inequality $z \geq 2 \eps$, implies
  $\langle \vctr{w}_j, \vctr{v}_{ij} \rangle \geq 2 \eps$
  for all $i \in [n] \setminus j$.
  Using the fact that $x_k,y_k \geq 0$ implies
  $x_k + y_k \geq |x_k - y_k|$, we see that the
  second dual constraint implies $\|\vctr{w}_j\|_1 \leq 1$.
  Thus, there exist vectors $\vctr{w}_1,\ldots,\vctr{w}_n$
  in the $L_1$ unit ball,
  such that for all $i \neq j$ in $[n]$,
  $\langle \vctr{w}_j, \vctr{v}_{ij} \rangle
  \geq 2 \eps$.
  Recalling that $\vctr{v}_{ji} = - \vctr{v}_{ij}$,
  we may combine the inequalities
  \[
    2 \eps \leq \langle \vctr{w}_j, \vctr{v}_{ij} \rangle, \qquad
    2 \eps \leq \langle \vctr{w}_i, \vctr{v}_{ji} \rangle
  \]
  to obtain
  \begin{equation} \label{eq:centered.1}
    4 \eps \, \leq \,
    \langle \vctr{w}_j - \vctr{w}_i, \vctr{v}_{ij} \rangle
    \, \leq \,
    \| \vctr{w}_j - \vctr{w}_i \|_1 \cdot \| \vctr{v}_{ij} \|_{\infty}
    \, = \,
    \| \vctr{w}_j - \vctr{w}_i \|_1
  \end{equation}
  Summarizing, the $L_1$ unit ball contains $n$ vectors
  $\vctr{w}_1,\ldots,\vctr{w}_n$, and the pairwise
  distances between these vectors (in $L_1$) are at
  least $4 \eps$. The $L_1$ balls of radius $2 \eps$
  centered at these vectors are pairwise disjoint,
  and all of them are contained in the ball of radius
  $1 + \eps$ centered at $\vctr{0}$, so the combined volume of
  the $n$ balls of radius $2\eps$ must not exceed the
  volume of the radius-$(1+\eps)$ ball. Since $1+\eps < 2$,
  this implies
  $(2 \eps)^m \cdot n < 2^m$,
  hence $n < (1/\eps)^m$ as claimed.
\end{proof}

\section{Achieving the optimal threshold prophet ratio}
\label{sec:pairwise}

Recall that when one is free to order the elements
arbitrarily, threshold stopping
rules satisfy a prophet inequality with factor
$1 - \frac1e$ but (asymptotically) no greater,
i.e.~$\tpratio(S_n) = 1 - \frac1e + o(1)$.
In this section we construct a small
set of permutations that achieves
this bound, and an even smaller set
that comes arbitrarily close. The
constructions make use of {\em pairwise
independence} and {\em almost pairwise independence};
see \Cref{defn:pairwise} below.

The fact that $\tpratio(S_n) = 1 - \frac1e + o(1)$
is implicit in \citep{blind-strat,delegated-search};
we prove the following
in \Cref{sec:pairwise-defer}
for the sake of making our exposition self-contained.
\begin{prop} \label{prop:1-1/e-optimal}
  As $n \to \infty$, the threshold prophet
  ratio $\tpratio(S_n)$ converges to $1 - \frac1e$
  from above.
\end{prop}

To design sets of permutations that achieve,
or approach, the $1 - \frac1e$ bound, we use
(almost) pairwise independent permutations,
a notion we now define.

\begin{defn} \label{defn:pairwise}
  A distribution over permutations $\sigma \in S_n$
  is {\em pairwise independent} if for every
  pair of distinct indices $i \neq j$ in $[n]$,
  the pair $(\sigma(i),\sigma(j))$ is distributed uniformly over
  $\{(a,b) \in [n] \times [n] \mid a \neq b \}$.
  It is {\em $(\eps,\delta)$-almost pairwise independent}
  if for every pair of distinct indices $i \neq j$,
  the distribution of the pair
  $\left( \left\lceil \frac{\sigma(i)}{\eps n} \right\rceil,
  \left\lceil \frac{\sigma(j)}{\eps n} \right\rceil \right)$
  is $\delta$-close, in total variation distance,
  to the uniform distribution on
  $[\frac{1}{\eps}] \times [\frac{1}{\eps}]$,
  where $[\frac{1}{\eps}]$ denotes the set
  $\{1,2,\ldots,\left\lceil \frac{1}{\eps} \right\rceil \}$.
\end{defn}

\begin{lem} \label{lem:pairwise-existence}
  For prime $n$, there exists a set $\Pi$
  of $n(n-1)$ permutations such that the
  uniform distribution over $\Pi$ is
  pairwise independent. For any
  $\eps,\delta > 0$ such that $1/\eps$
  is an integer, if $n$ is an integer multiple
  of $1/\eps$ and $\eps n \ge 2/\delta$,
  then there exists
  a set $\Pi$ of $O((\frac{1}{\delta \eps})^{2} \log n)$
  permutations such that the uniform distribution
  over $\Pi$ is $(\eps,\delta)$-almost pairwise
  independent.
\end{lem}
The proof is deferred to \Cref{sec:pairwise-defer}.
The first part is proven using the permutations
$\sigma(k) = ak  + b \pmod{n}$ for all
$a \in [n-1]$ and $b \in [n]$. The second
part is proven using the probabilistic method
to show that a random $\Pi \subset S_n$ of
the given cardinality has positive probability
of being $(\eps,\delta)$-almost pairwise
independent. Explicit constructions using
$\epsilon$-biased sets \citep{naor93,tashma17}
can achieve
$|\Pi| = O((\frac{1}{\delta \eps})^{2+o(1)} \log n)$.

\begin{thm} \label{thm:pairwise}
  Suppose $\sigma$ is a random permutation of $[n]$
  and $\pi = \sigma^{-1}$. If $\sigma$
  is drawn from a pairwise independent
  distribution then there exists a
  threshold $\theta$ such that
  \begin{equation} \label{eq:1-1/e}
    \expect_{\pi} \left[ \expect X_{\pi,\theta} \right]
    \geq
    \left( 1 - \tfrac1e \right) \cdot \expect X_* .
  \end{equation}
  If $\sigma$ is drawn from an $(\eps,\eps^2)$-almost
  pairwise independent distribution then there exists a
  threshold $\theta$ such that
  \begin{equation} \label{eq:1-1/e+eps}
    \expect_{\pi} \left[ \expect X_{\pi,\theta} \right]
    \geq
    \left( 1 - \tfrac1e - O(\eps) \right) \cdot \expect X_* .
  \end{equation}
\end{thm}
\begin{proof}
  For a given threshold $\theta$, let $p = \Pr(X_* \ge \theta)$.
  We have
  \begin{equation} \label{eq:pw.1}
    \expect X_* \le \theta + \expect[(X_*-\theta)^+] \le
      \theta + \sum_{i=1}^n \expect[(X_i - \theta)^+] .
  \end{equation}
  For the random variable $X_\tau = X_{\pi,\theta}$, we have
  \begin{equation} \label{eq:pw.2}
    \expect X_\tau = p \theta + \expect[(X_\tau - \theta)^+]
      = p \theta + \sum_{i=1}^n c_i \expect[(X_i - \theta)^+]
  \end{equation}
  where
  \begin{equation} \label{eq:pw.3}
    c_i = \sum_{k=1}^n \Pr(\pi(k) = i) \cdot
      \prod_{\ell=1}^{k-1} \Pr(X_{\pi(\ell)} < \theta \mid \pi(k) = i)
  \end{equation}
  denotes the probability that no element is
  selected before the stopping rule observes
  $X_i$. We can bound the product
  occurring in the formula for $c_i$ using the
  arithmetic mean-geometric mean inequality,
  similarly to the proof of \Cref{thm:golden-is-achievable}.
  Let $q_j = \Pr(X_j < \theta)$ for
  each $j \in [n]$. For any set $S \subset [n]$,
  let $p_{k,i}(S) = \Pr(\pi([k-1]) = S \mid \pi(k) = i)$
  denote the conditional probability that $S$ is
  exactly equal to the set of elements observed before
  $X_i$, given that $\pi(k)=i$.
  \begin{equation}
    \label{eq:pw.4}
      \prod_{\ell=1}^{k-1}
      \Pr(X_{\pi(\ell)} < \theta \mid \pi(k) = i)
      =
      \sum_{S \subset [n]}
      p_{k,i}(S)
      \prod_{j \in S} q_j
      \; \geq \;
      \prod_{S \subset [n]}
      \left( \prod_{j \in S} q_j \right)^{p_{k,i}(S)}
       =
      \prod_{j \in [n] \setminus \{i\}} q_j^{\sum_{S \subset [n], j \in S} p_{k,i}(S)}
      .
  \end{equation}
  The exponent $\sum_{S \subset [n], j \in S} p_{k,i}(S)$
  on the right side is equal to $\Pr( \sigma(j) < k \mid \sigma(i) = k )$.
  Hence,
  \begin{equation} \label{eq:pw.5}
    c_i \geq \sum_{k \in [n]} \Pr(\pi(k) = i)
      \prod_{j \neq i} q_j^{\Pr( \sigma(j) < k \mid \sigma(i) = k )} .
  \end{equation}
  If $\sigma$
  is pairwise independent then $\Pr(\pi(k)=i) = \frac1n$
  and for all $j$,
  $\Pr(\sigma(j) < k \mid \sigma(i) = k) = \frac{k-1}{n-1}$.
  Hence, if we let $q = \prod_{j=1}^n q_j$, then
  \begin{equation} \label{eq:pw.6}
    c_i \geq \frac1n \sum_{k=1}^n \left( \prod_{j \neq i} q_j \right)^{\frac{k-1}{n-1}}
        \geq \frac1n \left( 1 + q^{\frac{1}{n-1}} + \cdots + q^{\frac{n-2}{n-1}} + q \right) .
  \end{equation}
  If we set $\theta$ so that $q = \frac1e$, then
  Lemma~\ref{lem:1/e} in \Cref{sec:pairwise-defer}
  shows that the right side of~\eqref{eq:pw.6}
  is greater than $1-\frac1e$. Also, note that
  $$ p = \Pr(X_* \geq \theta) = 1 - \Pr(X_* < \theta)
     = 1 - \prod_{j=1}^n q_j = 1 - q  = 1 - \tfrac1e . $$
  Having shown that $p = 1 - \tfrac1e$ and that
  $c_i > 1 - \tfrac1e$ for all $i$, we may substitute
  these bounds into~\eqref{eq:pw.2} and conclude that
  \begin{equation} \label{eq:pw.7}
    \expect X_\tau \geq \left(1-\tfrac1e\right) \theta + \left(1 - \tfrac1e\right)
    \sum_{i=1}^n \expect[(X_i - \theta)^+] \geq \left(1 - \tfrac1e\right) \expect X_* .
  \end{equation}

  Now we turn to the case that the distribution of
  $\sigma$ is $(\eps,\delta)$-almost pairwise independent
  for $\delta = \eps^2$.
  In that case, we group the time steps $k \in [n]$ into
  ``buckets'' of $\eps n$ consecutive steps; the bucket
  containing the time step when $X_i$ is observed
  is numbered
  $b(i) = \left\lceil \frac{\sigma(i)}{\eps n} \right\rceil$.
  The counterpart of \eqref{eq:pw.3} is the following
  inequality:
  \begin{equation} \label{eq:pw.8}
    c_i \geq \sum_{u = 1}^{1/\eps} \Pr(b(i) = u) \cdot
      \prod_{\ell=1}^{\eps n u}
      \Pr(X_{\pi(\ell)} < \theta \mid b(i) = u)
  \end{equation}
  The inequality is valid because the left side is the
  probability that no value greater than $\theta$ is
  observed before $X_i$ is observed, and the right side
  is the probability no value greater than $\theta$ is
  observed in buckets $1,2,\ldots,b(i)$.
  For any set $S \subset [n]$,
  let $\tilde{p}_{u,i}(S) = \Pr(\pi([\eps n u]) = S \mid b(i) = u)$
  denote the conditional probability that $S$ is
  exactly equal to the set of elements observed in
  buckets $1,2,\ldots,u$, given that $b(i)=u$.
  Analogously to \eqref{eq:pw.4}, we can use
  the AM-GM inequality to derive
  \begin{equation}
    \label{eq:pw.9}
      \prod_{\ell=1}^{\eps n u}
      \Pr(X_{\pi(\ell)} < \theta \mid b(i)=u)
      \; \geq \;
      \prod_{j \in [n]} q_j^{\sum_{S \subset [n], j \in S} \tilde{p}_{u,i}(S)}
      .
  \end{equation}
  The exponent of $q_j$ on the right side is equal
  to $\Pr(b(j) \le u \mid b(i) = u)$. If the
  distribution of $\sigma$ is $(\eps,\delta)$-almost
  pairwise independent, this probability is
  at most $\eps u + \delta$ and $\Pr(b(i) = u) \geq \eps-\delta$.
  Therefore,
  \begin{equation} \label{eq:pw.10}
    c_i \geq (\eps - \delta) \cdot \sum_{u=1}^{1/\eps} q^{\eps u + \delta}
        = \left(1 - \tfrac{\delta}{\eps} \right)
        \cdot q^{\eps+\delta} \cdot
        \left[ \eps \left(
          1 + q^{\eps} + q^{2 \eps} + \cdots + q^{1 - \eps}
        \right) \right] .
  \end{equation}
  If $q = \frac1e$ and $\delta = \eps^2$ then the factors
  $1 - \frac{\delta}{\eps}$ and $q^{\eps+\delta}$ are
  both $1 - O(\eps)$ and the factor
  $\eps \left(
    1 + q^{\eps} + q^{2 \eps} + \cdots + q^{1 - \eps}
  \right)$ is at least $1 - \frac1e$,
  again by \Cref{lem:1/e}. Hence, $c_i \geq 1 - \frac1e - O(\eps)$.
  As before, $p = 1-q = 1 - \frac1e$, and the
  lemma follows by substituting these bounds for $p$ and $c_i$
  into~\eqref{eq:pw.2} and comparing with~\eqref{eq:pw.1}.
\end{proof}

We now restate and prove \Cref{thm:pseudorandom}.
\begin{thm}
  For every $n \in \mathbb{N}$ and
  $\eps>0$, there is a set $\Pi$ consisting
  of $O(\poly(\eps^{-1}) \cdot \log n)$ permutations
  such that $\tpratio(\Pi) > 1 - \frac1e - \eps$.
  There is also a set $\Pi$ consisting of
  $O(n^2)$ permutations such that $\tpratio(\Pi) \geq
  1 - \frac1e$.
\end{thm}
\begin{proof}
    If $n$ is prime, \Cref{lem:pairwise-existence} combined
    with \Cref{thm:pairwise} show that there exists a
    set $\Pi$ of $n (n-1)$ permutations such that
    $\tpratio(\Pi) \geq 1 - \frac1e$. If $n$ is composite,
    we make use of a ``padding lemma''
    (\Cref{lem:padding} in \Cref{sec:pairwise-defer})
    that says that for any $N \geq n$,
    a set of $m$ permutations $\Pi_N \subseteq S_N$
    can be transformed into a set of at most $m$ permutations,
    $\Pi_n \subseteq S_n$,
    such that $\tpratio(\Pi_n) \geq \tpratio(\Pi_N)$.
    Taking $N$ to be a prime between $n$ and $2n$
    \citep{cheby}, the padding lemma implies there
    is a set $\Pi_n$ of fewer than $4n^2$ permutations
    satisfying $\tpratio(\Pi) \geq 1 - \frac1e$.
    Similarly, for $\eps>0$,
    take $k \in \naturals$
    such that $\eps/2 < \frac1k \leq \eps$.
    If $2 \leq n \leq 2k^3$ then
    $n^2 \leq 4 k^6 \leq 256 \eps^{-6}$, and we have already
    shown the existence of $\Pi \subseteq S_n$ with
    $|\Pi| = O(n^2) = O(\eps^{-6})$
    and $\tpratio(\Pi) \geq 1 - \frac1e$.
    Otherwise, $n > 2 k^3$ so there is a multiple of $k$
    between $n$ and $2n$. Denoting this number by $N$,
    and observing that the constraint $\eps' N \ge 2 / \delta'$
    is satisfied when $\eps' = 1/k$ and $\delta' = (\eps')^2$,
    we apply \Cref{lem:pairwise-existence} and \Cref{thm:pairwise}
    to deduce the existence of a set $\Pi_N$ of $O(k^6 \log N)$
    permutations of $[N]$ such that $\tpratio(\Pi_N) \geq 1 - \frac1e - \eps$,
    then we use the padding lemma to deduce the existence
    of $\Pi_n \subseteq S_n$ such that $|\Pi_n| = O(k^6 \log N)
    = O(\eps^{-6} \log n)$ and $\tpratio(\Pi_n) \geq 1- \frac1e - \eps$.
\end{proof}

\appendix

\section{Deferred proofs from \Cref{sec:beating}}
\label{sec:beating-defer}

In this section we restate and prove
\Cref{lem:golden-is-optimal}.

\begin{lem}
  If $\Pi$ is a non-empty set of permutations of $[n]$
  and there exists an index $j \in [n]$ that is
  $\eps$-centered with respect to $\Pi$, then
  $\tpratio(\Pi) \leq \golden^{-1} + O(\eps)$.
\end{lem}
\begin{proof}
  Suppose $j$ is $\eps$-centered with respect to $\Pi$,
  and let $p$ be a probability distribution on
  $[n] \setminus \{j\}$ such that for
  every permutation $\sigma$ whose inverse
  belongs to $\Pi$, the sets
  $\{ i \mid \sigma(i) < \sigma(j)\}$ and
  $\{ i \mid \sigma(i) > \sigma(j)\}$ have
  measure at least $\frac12 - \eps$ under $p$.
  Define the distributions of $X_1,X_2,\ldots,X_n$
  as follows. For a small positive number $\delta$ to be
  determined later, the value of $X_j$ is
  $(\sqrt{5}-1)/\delta$ with probability $\delta$,
  and otherwise $X_j = 0$. For every
  $i \in [n] \setminus \{j\}$,
  $X_i$ has cumulative distribution function
  $F_i$ satisfying
  \[
    F_i(1-t) = \Pr(X_i \le 1-t) = \begin{cases}
      1 & \mbox{if } t < 0 \\
      \exp \left( -\frac{p(i) t}{\delta} \right) & \mbox{if } 0 \le t \le 1 \\
      0 & \mbox{if } t > 1 .
    \end{cases}
  \]
  In other words, $1 - X_i = Y_i \wedge 1$ where
  $Y_i$ is exponentially distributed with rate parameter
  $\frac{p(i)}{\delta}$ and the notation $a \wedge b$ denotes
  the minimum of $a$ and $b$.

  First, observe that $X_*$ is equal to $(\sqrt{5} - 1)/\delta$
  with probability $\delta$, and otherwise
  \begin{align*}
    X_* = \max_{i \neq j} \{X_i\}
        = 1 - \min_{i \neq j} \{ Y_i \wedge 1 \}
        \ge 1 -  \bigwedge_{i \neq j} Y_i  .
  \end{align*}
  The minimum of independent exponential random variables
  with rates $r_1,\ldots,r_n$ is exponential with rate
  $r_1 + \cdots + r_n$. Hence, $ \bigwedge_{i \neq j} Y_i$
  is exponentially distributed with rate
  $\frac{1}{\delta} \sum_{i \neq j} p(i) = \frac{1}{\delta}$,
  and its expected value is $\delta$. Consequently,
  \[
    \expect [ X_* \mid X_j = 0 ] \ge
    \expect \left[ 1 - \bigwedge_{i \neq j} Y_i \right] = 1 - \delta
  \]
  and the prophet's unconditional expected value satisfies
  \begin{equation} \label{eq:grio.1}
    \expect X_* = \delta \cdot \frac{\sqrt{5}-1}{\delta} \, + \,
      (1 - \delta) \cdot \expect [ X_* \mid X_j = 0 ]
      \ge \sqrt{5} - 1 + (1 - \delta)^2 > \sqrt{5} - 2 \delta.
  \end{equation}
  Now we turn to analyzing the expected value obtained by a
  threshold stopping rule $\tau$ with threshold $\theta$,
  assuming $\tau$ is $\pi$-adapted for some $\pi \in \Pi$
  with inverse permutation $\sigma = \pi^{-1}$.
  If $\theta < 0$ then $\tau = 1$ and $X_\tau = X_{\pi(1)} \le 1$
  so $\expect X_\tau \le 1$. If $\theta > 1$ then
  $\tau \in \{j, \bot\}$ and $X_\tau = X_j$, so
  $\expect X_\tau = \expect X_j = \sqrt{5} - 1$.
  In both of these cases, $\expect X_\tau < \golden^{-1}
  \cdot \expect X_*$ provided $\delta < 0.07$.
  The remaining case to consider is when
  $0 \le \theta \le 1$. In that case let
  $I_0 = \{ i \mid \sigma(i) < \sigma(j)\}$
  denote the set indexing the values appearing
  before $X_j$ in the sequence
  $X_{\pi(1)},\ldots,X_{\pi(n)}$,
  and let $I_1 = \{i \mid \sigma(i) > \sigma(j) \}$
  denote the set indexing the values appearing
  after $X_j$ in that sequence.
  If $q_0, q_1$ denote the probabilities
  of the events $\max \{ X_i \mid i \in I_0 \} < \theta$
  and $\max \{ X_i \mid i \in I_1 \} < \theta$,
  respectively, then the gambler's expected value
  satisfies
  \begin{equation} \label{eq:grio.2}
      \expect X_\tau \; \le \;
      (1 - q_0) \cdot 1 \, + \,
      q_0 \cdot \delta \cdot \frac{\sqrt{5}-1}{\delta} \, + \,
      q_0 \cdot (1-\delta) \cdot (1-q_1) \cdot 1
      \; < \; 1 - q_0 q_1 + (\sqrt{5} - 1) q_0 ,
  \end{equation}
  where the first inequality is justified because
  $X_\tau \le 1$ if $\tau = i$ for any $i \in [n] \setminus j$
  and $X_\tau = 0$ if $\tau = \bot$.
  Writing $\theta = 1 - \delta s$, we have
  $\Pr(X_i < \theta) = \exp(- s p(i))$ for $i \neq j$
  and therefore
  \begin{align*}
    q_0 & = \prod_{i \in I_0} \Pr(X_i < \theta) =
      \exp \left( - s \sum_{i \in I_0} p(i) \right) <
      \exp \left( - (\tfrac12 - \eps) s \right) \\
    q_0 q_1 &= \prod_{i \in I_0 \cup I_1} \Pr(X_i < \theta) =
      \exp \left( - s \sum_{i \neq j} p(i) \right) = \exp(-s) .
  \end{align*}
  Substituting these into~\eqref{eq:grio.2} we find that
  \begin{equation} \label{eq:grio.3}
    \expect X_\tau < 1 - \exp(-s) + (\sqrt{5}-1) \exp \left( - (\tfrac12 - \eps) s \right) .
  \end{equation}
  To complete the proof we must bound the right
  side of~\eqref{eq:grio.3} above by $\golden^{-1} + O(\eps)$
  when $s \ge 0$. The
  derivative of the right side is
  \[
    \exp(-s) - ( \sqrt{5}-1 ) ( \tfrac12 - \eps )
    \exp \left( - (\tfrac12 - \eps) s \right)
      =
    \left[ \exp \left(-(\tfrac12 + \eps) s \right)
    - \golden^{-1} + (\sqrt{5} - 1) \eps \right]
    \cdot \exp \left( - (\tfrac12 - \eps) s \right)
  \]
  which is negative for all $s \ge 1$, provided that
  $\eps < \frac{1}{110}$. Assume henceforth that
  $\eps < \frac{1}{110}$, as otherwise the inequality
  asserted by the lemma, $\tpratio(\Pi) \leq \golden^{-1} + O(\eps)$,
  is trivially satisfied due to the $O(\eps)$ term.
  The derivative calculation above shows
  that the right side of~\eqref{eq:grio.3}
  is a decreasing function of $s \ge 1$, implying that its
  maximum value on the interval $s \in [0,\infty)$ is
  attained when $s \in [0,1]$. For $s$ in this range,
  if we let $q = \exp( - \frac12 s)$, then
  \[
    \exp( - (\tfrac12 - \eps) s) =
    q \cdot e^{\eps s} \le q \cdot e^{\eps} \le q + O(\eps) .
  \]
  Hence we can rewrite~\eqref{eq:grio.3} as
  \begin{equation} \label{eq:grio.4}
    \expect X_\tau < 1 - q^2 + (\sqrt{5}-1) q + O(\eps) .
  \end{equation}
  The right side of~\eqref{eq:grio.4} is a quadratic
  function of $q$ that is maximized when $q = \golden^{-1}$
  and $1 - q^2 = q = \golden^{-1}$. Therefore,
  \begin{equation} \label{eq:grio.5}
    \expect X_\tau < \sqrt{5} \golden^{-1} + O(\eps) .
  \end{equation}
  Combining \eqref{eq:grio.1} with \eqref{eq:grio.5},
  the bound
  $\expect X_\tau \leq (\golden^{-1} + O(\eps)) \cdot \expect X_*$
  follows, which confirms that
  $\tpratio(\Pi) \leq \golden^{-1} + O(\eps)$.
\end{proof}

\section{Deferred proofs from \Cref{sec:pairwise}}
\label{sec:pairwise-defer}

This section contains proofs of propositions and
lemmas that were mentioned in \Cref{sec:pairwise}
whose proofs were deferred.

\begin{prop}
  As $n \to \infty$, the threshold prophet
  ratio $\tpratio(S_n)$ converges to $1 - \frac1e$
  from above.
\end{prop}
\begin{proof}
  By \Cref{thm:pseudorandom}, $\tpratio \ge 1 - \frac1e$,
  so we need only show that $\tpratio \le 1 - \frac1e + o(1)$
  as $n \to \infty$. To this end, let $H \gg 1$ be an arbitrarily
  large number and consider a sequence of
  random variables $X_1,\ldots,X_n$ drawn i.i.d.~from a
  distribution whose cumulative distribution function
  $F$ is given by
  \begin{equation} \label{eq:cdf}
    F(x) = \begin{cases}
      0 & \mbox{if } x < 1 \\
      \left( H - \tfrac{1}{(e-2)n} \right) (x-1)
      & \mbox{if } 1 \leq x \leq 1 + \tfrac1H \\
      1 - \frac{1}{(e-2)nH}
      & \mbox{if } 1 + \tfrac1H < x < H+1 \\
      1 - \tfrac{1}{(e-2)nH} + \tfrac{x - H - 1}{(e-2)n}
      & \mbox{if } H+1  \leq x \leq H+1 + \tfrac1H \\
      1 & \mbox{if } x > H+1+\tfrac1H .
    \end{cases}
  \end{equation}
  In words, each $X_i$ is sampled from a
  mixture of two uniform distributions, on
  the intervals $[1,1+\frac{1}{H}]$ and
  $[H+1,H+1+\frac{1}{H}]$, with the mixture
  weights being $1 - \frac{1}{(e-2)nH}$ and
  $\frac{1}{(e-2)nH}$, respectively.

  Since the variables $X_1,\ldots,X_n$ are identically
  distributed, reordering them has no effect on the
  performance of stopping rules, so we merely need to
  show that if $\tau$ is a threshold stopping rule
  adapted to the sequence $X_1,\ldots,X_n$ then
  $\expect X_\tau \leq (1 - \frac1e + o(1)) \cdot \expect X_*$,
  where the $o(1)$ term vanishes as $n \to \infty$.

  Let $p = \frac{1}{(e-2)nH}$ denote the
  probability that a sample $X_i$ exceeds $H$.
  The prophet's expected value satisfies
  \begin{equation} \label{eq:p1eo.1}
      \expect X_* = 1 + (1 - (1-p)^n) H + O(\tfrac1H)
         \ge \frac{e-1}{e-2} - O(\tfrac1H)
  \end{equation}
  where the second inequality follows from:
  \begin{align*}
    (1-p)^n & < 1 - np + \binom{n}{2} p^2 <
    1 - \frac{1}{(e-2)H} + \frac{1}{2 (e-2)^2 H^2} \\
    (1 - (1-p)^n) H & > \frac{1}{e-2} - \frac{1}{2 (e-2)^2 H} .
  \end{align*}
  Now consider a threshold stopping rule $\tau$ that
  sets a threshold $\theta$ such that $F(\theta) = 1 - q$;
  in other words, any given element of the sequence
  $X_1,\ldots,X_n$ has probability $q$ of exceeding the
  threshold. Then, $\tau = \bot$
  with probability $(1-q)^n$. Furthermore,
  conditional on $\tau < \bot$, the events
  $X_\tau \in [H + 1, H + 1 + \frac1H]$
  and $X_\tau \in [1, 1 + \frac1H]$ have
  conditional probabilities $\min \{ p/q, 1 \}$
  and $1 - \min \{p/q, 1\},$ respectively.
  Thus,
  \[
    \expect X_\tau \le \left( 1 - (1-q)^n \right) \cdot
    \left( \tfrac{p}{q} H + 1 + \tfrac1H \right) .
  \]
  Let $r = q/n$.
  Ignoring $O(\tfrac1H)$ terms that vanish as $H \to \infty$
  and $O(1/n)$ terms that vanish as $n \to \infty$ we have
  \begin{align}
  \nonumber
    \expect X_{\tau} & \approx \left(1 - e^{-r} \right) \cdot
    \left(1 + \tfrac{1}{(e-2)r} \right) \\
  \label{eq:1-1/e.0}
    \frac{\expect X_\tau}{\expect X_*} & \approx \left(1 - e^{-r} \right) \cdot \left( \frac{e-2}{e-1} + \frac{1}{(e-1)r} \right).
  \end{align}
  The right side of~\eqref{eq:1-1/e.0} is maximized at $r=1$, where it equals $1 - \frac1e$.
\end{proof}

\begin{lem}
  For prime $n$, there exists a set $\Pi$
  of $n(n-1)$ permutations such that the
  uniform distribution over $\Pi$ is
  pairwise independent. For any
  $\eps,\delta > 0$ such that $1/\eps$
  is an integer, if $n$ is an integer multiple
  of $1/\eps$ and $\eps n \ge 2/\delta$,
  then there exists
  a set $\Pi$ of $O((\frac{1}{\delta \eps})^{2} \log n)$
  permutations such that the uniform distribution
  over $\Pi$ is $(\eps,\delta)$-almost pairwise
  independent.
\end{lem}
\begin{proof}
  If $n$ is equal to a prime number $p$,
  for any integers $a,b$ such
  that $a$ is not divisible by $p$,
  the function $x \mapsto ax + b \pmod{p}$
  is a permutation of $[p]$. If $(a,b)$ are
  sampled uniformly at random from
  $[p-1] \times [p]$, this defines
  a pairwise independent permutation
  distribution. The reason is that
  for any $(i,j)$ and $(k,\ell)$ in
  $[p] \times [p]$ such that $i \neq j, \,
  k \neq \ell$, the system of linear congruences
  $$ ai+b \equiv k, \quad aj+b \equiv \ell \pmod{p} $$
  has the unique solution $a \equiv (k-\ell)(i-j)^{-1},
  \, b \equiv k - ai \pmod{p}$.

  For $\eps,\delta > 0$ such that $1/\eps$ and $\eps n$
  are integers and $\eps n \geq 2/\delta$, we use the
  probabilistic method to prove the existence of a
  multiset $\Pi$ of $m = O((\eps \delta)^{-2} \log n)$
  permutations such that the uniform distribution
  on $\Pi$ is $(\eps,\delta)$-almost pairwise independent.
  In fact, we will prove that the multiset obtained by
  drawing $m$ uniformly-random samples, with replacement,
  satisfies this property with positive probability.
  Define the function
  $b(k) = \left\lceil \frac{k}{\eps n} \right\rceil$,
  mapping $[n]$ to $[\frac{1}{\eps}]$.
  The $(\eps,\delta)$-almost pairwise independence
  property asserts that when $\sigma$ is a random
  sample from the distribution, for any $i \neq j$
  the distribution of the pair $(b(\sigma(i)),b(\sigma(j)))$ is
  is $\delta$-close to the uniform distribution on
  $[\frac{1}{\eps}] \times [\frac{1}{\eps}]$
  in total variation distance.
  For any pair $(u,v) \in [\frac{1}{\eps}] \times [\frac{1}{\eps}]$,
  a uniformly random $\sigma \in S_n$ satisfies:
  \begin{equation} \label{eq:pe.1}
    \Pr \big( (b(\sigma(i)),b(\sigma(j))) = (u,v) \big) \; = \;
    \begin{cases}
      \frac{\eps n}{n} \cdot \frac{\eps n - 1}{n-1} & \mbox{if } u=v \\
      \frac{\eps n}{n} \cdot \frac{\eps n}{n-1} & \mbox{if } u \neq v .
    \end{cases}
  \end{equation}
  In both cases we have
  \begin{equation} \label{eq:pe.2}
    \Pr \big( (b(\sigma(i)),b(\sigma(j))) = (u,v) \big)  >
    \frac{\eps n}{n} \cdot \frac{\eps n - 1}{n} =
    \eps \left( \eps - \frac1n \right) \geq
    \eps^2 \left( 1 - \frac{\delta}{2} \right)
  \end{equation}
  where the last inequality is a consequence of
  $\eps n \geq 2 / \delta$.

  Now, suppose $\sigma_1,\ldots,\sigma_m$ are i.i.d.\ random
  draws from the uniform distribution on $S_n$. Fix an
  arbitrary $i \neq j$ in $[n]$ and an arbitrary pair of
  buckets $(u,v) \in [\frac{1}{\eps}] \times [\frac{1}{\eps}]$.
  For $s = 1,2,\ldots,m$ let
  \begin{align*}
    Y_s & = \begin{cases}
      1 & \mbox{if } (b(\sigma(i)),b(\sigma(j))) = (u,v) \\
      0 & \mbox{otherwise}
    \end{cases} \\
    Z_s & = (\expect Y_s) - Y_s .
  \end{align*}
  In light of~\eqref{eq:pe.2}, we have $\expect \left[ \sum_{s=1}^m Y_s \right] >
  \eps^2 \left( 1 - \frac{\delta}{2} \right) m$, so the
  event $\sum_{s=1}^m Y_s \leq \eps^2 (1-\delta) m$ only
  happens when $\sum_{s=1}^m Z_s > \frac12 m \eps^2 \delta$.
  Bernstein's Inequality~\citep{bernstein} ensures that
  \begin{equation} \label{eq:pe.bern}
    \Pr \left( \sum_{s=1}^m Z_s > \tfrac12 m \eps^2 \delta \right) \; < \;
    \exp \left( - \frac{\tfrac18 m^2 \eps^4 \delta^2}{m \eps^2 + \tfrac16 m \eps^2 \delta} \right) \; = \;
    \exp \left( - \frac{m \eps^2 \delta^2}{8 (1 + \delta/6)} \right) .
  \end{equation}
  Since $\delta < 1$, the right side can be made less than $(\eps/n)^2$
  by setting $m \ge 18 (\eps \delta)^{-2} \log(n/\eps) .$
  Note that our assumption that $n$ is divisible by $1/\eps$
  implies $n / \eps \le n^2,$ hence $\log(n/\eps) \le 2 \log n$.
  Thus, drawing $m = 36 (\eps \delta)^{-2} \log n$ samples
  suffices to make
  $\Pr \left( \sum_{s=1}^m Z_s > \tfrac12 m \eps^2 \delta \right)$
  less than $(\eps/n)^2.$

  For $(i,j,u,v) \in [n] \times [n] \times [\frac{1}{\eps}] \times
  [\frac{1}{\eps}]$ with $i \neq j$, define $S(i,j,u,v) \subset S_n$
  to be the set of permutations $\sigma$ such that
  $(b(\sigma(i)),b(\sigma(j))) = (u,v)$.
  Let $\mathcal{E}(i,j,u,v)$ denote the event that
  $S(i,j,u,v)$ contains $m \eps^2 (1-\delta)$ or fewer
  of the permutations $\sigma_1,\ldots,\sigma_m$, and
  let $\mathcal{E} = \bigcup_{i,j,u,v} \mathcal{E}(i,j,u,v)$.
  Our calculation using Bernstein's Inequality showed that
  for $m \geq 36 (\eps \delta)^{-2} \log n$, we have
  $\Pr(\mathcal{E}(i,j,u,v)) < (\eps/n)^2$.
  Taking the union bound over all $n(n-1)/\eps^2$ choices
  of $i,j,u,v$, we conclude that $\Pr(\mathcal{E}) < 1$.
  Therefore, the complementary event $\overline{\mathcal{E}}$ occurs
  with positive probability. When  $\overline{\mathcal{E}}$ occurs,
  we claim the uniform distribution over
  $\{\sigma_1,\ldots,\sigma_m\}$ is $(\eps,\delta)$-almost
  pairwise independent. To verify this, consider any $i \neq j$
  and let $D_{ij}$ denote the distribution of $(b(\sigma(i)),b(\sigma(j)))$
  when $\sigma$ is drawn randomly from $\{\sigma_1,\ldots,\sigma_m\}$.
  Let $U$ denote the uniform distribution on $[\frac{1}{\eps}] \times
  [\frac{1}{\eps}]$. We have
  \begin{align*}
      \| D_{ij} - U \|_{\mathrm{TV}} &=
      \sum_{u=1}^{1/\eps} \sum_{v=1}^{1/\eps}
        \big( U(S(i,j,u,v)) -
              D_{ij}(S(i,j,u,v)) \big)^+
        \; < \;
        \left( \tfrac{1}{\eps} \right)^2
        \left( \eps^2 - \frac{m\eps^2(1-\delta)}{m} \right)
        \; = \;
        \delta
  \end{align*}
  as required by the definition of $(\eps,\delta)$-almost
  pairwise independence.
\end{proof}

\begin{lem} \label{lem:1/e}
  If $k \in \naturals$ and $r^k \geq \frac1e$ then
  $\frac{1}{k+1} \left( 1 + r + \cdots + r^k \right) \geq 1 -  \frac1e$.
\end{lem}
\begin{proof}
  The left side is an increasing function of $r$, so
  it suffices to prove the inequality when $r^k = \frac1e$.
  We begin by noting that if the hypothesis
  had been $r^{k+1} = \frac1e$ rather than $r^k = \frac1e$,
  the lemma would follow easily by comparing the sum to an integral.
  \[
    \frac{1}{k+1} \left( 1 + r + \cdots + r^k \right) >
    \frac{1}{k+1} \int_0^{k+1} r^t \, dt =
    \frac{1}{(k+1) \ln r} \left[ 1 - r^{k+1} \right] = 1 - \frac1e .
  \]
  We do not know of any comparably simple argument that
  works when $r^k = \frac1e$. Instead, we define the
  function
  \[
    f(k) = \frac{1}{k+1} \frac{ 1 - e^{-\frac{k+1}{k}} }{ 1 -
          e^{-\frac{1}{k}} }
  \]
  and argue that $f(k) \geq 1-\frac1e$ for all $k \geq 1$,
  by proving that $f(k)$ is monotonically decreasing in $k$
  and that $\lim_{k \to \infty} f(k) = 1 - \frac1e$.
  The computation of $\lim_{k \to \infty} f(k)$ follows
  from the facts that $1 - e^{-\frac{k}{k+1}} \to 1 - \frac1e$
  and $(k+1)(1 - e^{-1/k}) \to 1$ as $k \to \infty$.
  The monotonicity claim follows by expanding the domain of
  $f$ from the natural numbers to the real numbers and taking
  the derivative with respect to $k$.
  \[
  \frac{d f}{d k} = \frac{-e^{\frac{k+2}{k}} k^2 - k^2 +
  e^{\frac{1}{k}} ( (k+1)(k-2) + 1 ) + e^{\frac{k+1}{k}} ( k^2 + k + 1 ) }{ e \left(
  e^{\frac{1}{k}} - 1 \right)^2 (k+1)^2 k^2 }
  \]
  The denominator is positive for all $k \geq 1$
  so it
  suffices to prove that the numerator in non-positive which
  is equivalent to showing:
  \begin{align*}
  \label{ineq:toshow}
  e^{\frac{1}{k}} ( (k+1)(k-2) + 1 + ek(k+1) + e ) &\leq
  k^2 \left( 1 + e^{\frac{k+2}{k}} \right),\\
  e^{-\frac{1}{k}} + e^{\frac{k+1}{k}} &\geq
  \frac{ (e+1) (k+1)^2 - (e + 3)(k+1) + (e+1) }{k^2}
  \end{align*}
  To prove the last inequality for all $k \geq 1$, we make use
  of the following quadratic lower bounds on $e^x$ which can
  be shown using elementary calculus\footnote{Both are Taylor
  expansion approximations of $e^x$ around $0$ and $1$
  respectivelly.}.
  \begin{align*}
  &e^x \geq 1 + x + \frac{1}{e} x^2, \text{ for all } x \in [-1, 0]\\
  &e^x \geq \frac{e}{2} (x^2 + 1) \text{ for all } x \geq 1
  \end{align*}
  Notice that $-1 \leq -\tfrac{1}{k} \leq 0$ and
  $\tfrac{k+1}{k} \geq 1$ for $k \geq 1$ so the above
  inequalities apply.
  \begin{align*}
  e^{-\frac{1}{k}} + e^{\frac{k+1}{k}}
  &\geq 1 - \frac{1}{k} + \frac{1}{e} \left( \frac{k+1}{k}
  \right)^2 + \frac{e}{2} \left( \left( \frac{k+1}{k}
  \right)^2 + 1 \right)\\
  &= \frac{ \left( 1 + \frac1e + e \right) (k+1)^2 - (e + 3)(k+1) +
  \frac{e}{2} + 1 }{ k^2 }\\
  &\geq \frac{ (1 + e) (k+1)^2 - (e+3)(k+1) + (e+1) }{ k^2 }
  \end{align*}
  where the last inequality follows from $\frac1e (k+1)^2 +
  \frac{e}{2} + 2 \geq e + 1$ for $k \geq 1$.
\end{proof}

\begin{lem}[Padding Lemma] \label{lem:padding}
  If $N \in \naturals$ and $\Pi_N \subseteq S_N$ is a set of $m$
  permutations such that $\tpratio(\Pi_N) \geq \alpha$,
  then for all $n \leq N$ there is a set $\Pi_n \subseteq S_n$
  of at most $m$ permutations such that $\tpratio(\Pi_n) \geq \alpha$.
\end{lem}
\begin{proof}
  For any permutation $\pi \in S_N$ let $\pi_{\downarrow n}$ denote the
  unique permutation of $[n]$ that can be expressed as the composition
  $\pi \circ f$ for some monotonically increasing $f : [n] \to [N]$.
  (This $f$ maps the elements of $[n]$ to the elements
  of $\pi^{-1}([n])$ in increasing order.)
  If $\Pi \subseteq S_N$ satisfies $|\Pi| = m$ and
  $\tpratio(\Pi) \geq \alpha$ then
  the set $\Pi_n = \{ \pi_{\downarrow n} \mid \pi \in \Pi \}$
  certainly satisfies $|\Pi_n| \leq m$; we claim it also satisfies
  $\tpratio(\Pi_n) \geq \alpha$. To see why, consider any
  independent non-negative-valued random variables
  $X_1,\ldots,X_n$. Extend this to a sequence $X_1,\ldots,X_N$
  by defining $X_i \equiv 0$ for $n < i \leq N$. Since
  $\tpratio(\Pi) \geq \alpha$ there is a $\pi \in \Pi$
  and a threshold $\theta$ such
  $\expect X_{\pi,\theta} \geq \alpha \cdot \expect X_*$.
  (Note that $\max \{X_1,\ldots,X_N\} = \max \{X_1,\ldots,X_n\}$
  so it is immaterial whether $X_*$ is interpreted as referring
  to the former or the latter quantity.) Let $\pi' = \pi_{\downarrow n}$,
  and note that $\pi' \in \Pi_n$. For any threshold $(\theta,\tth) \in
  \reals \times [0,1]$ the stopping rule $\tau(\pi,\theta,\tth)$
  stops at the earliest element in the sequence
  $X_{\pi(1)},X_{\pi(2)},\ldots,X_{\pi(n)}$ such that
  $(X_{\pi(i)},\tX_{\pi(i)}) \ge (\theta,\tth)$, whereas
  $\tau(\pi',\theta,\tth)$ stops at the earliest element
  such that $(X_{\pi(i)},\tX_{\pi(i)}) \ge (\theta,\tth)$
  {\em and} $\pi(i) \in [n]$. These two elements are the
  same unless $\theta=0$ and $X_{\pi,\theta} = 0$, because
  when $\pi(i) \not\in [n]$ we have $X_{\pi(i)} = 0$ by
  construction. Therefore, $X_{\pi',\theta} \geq X_{\pi,\theta}$
  pointwise, and we find that
  \[
    \expect X_{\pi',\theta} \geq \expect X_{\pi,\theta} \geq
    \alpha \cdot \expect X_*,
  \]
  as desired.
\end{proof}

\bibliographystyle{apalike}
\bibliography{prophet}

\end{document}